\documentclass[submission,copyright,creativecommons]{eptcs}
\providecommand{\event}{TARK XVII: Theoretical Aspects of Rationality and Knowledge 2019} 
\usepackage{underscore} 
\usepackage[utf8]{inputenc}

\usepackage{amsmath}
\usepackage{amssymb}
\usepackage{microtype}
\usepackage{etex}
\usepackage{color}
\usepackage{booktabs}
\usepackage{enumerate}
\usepackage{marvosym}
\usepackage{float}
\usepackage{tikz}
\usetikzlibrary{backgrounds,positioning,trees,shapes,arrows.meta}

\usepackage{amsthm}
\newtheorem{theorem}{Theorem}
\newtheorem{lemma}{Lemma}
\newtheorem{definition}{Definition}
\newtheorem{example}{Example}

\definecolor{myRed}{rgb}{0.75,0,0}

\renewcommand{\phi}{\varphi}

\newcommand{\imp}{\rightarrow}
\newcommand{\lra}{\leftrightarrow}

\newcommand{\Lng}{\mathcal{L}}
\newcommand{\Prop}{\mathsf{Prop}}

\newcommand{\SMLD}{\mathbb{SMLD}}
\newcommand{\SPALD}{\mathbb{SPALD}}

\newcommand{\Kv}{\ensuremath{\textit{Kv}}}
\newcommand{\Kd}{\ensuremath{\textit{Kd}}}
\newcommand{\Kx}{\ensuremath{\textit{Kx}}}
\newcommand{\M}{\mathcal{M}}
\newcommand{\DF}{\ensuremath{\textsf{DEF}}}
\newcommand{\df}{\ensuremath{\textsf{def}}}

\newcommand{\mrg}{\mathsf{merge}}
\newcommand{\pick}{\mathsf{pick}}

\newcommand{\col}{\mathopen:}

\title{How to Agree without Understanding Each Other:\texorpdfstring{\\}{ }Public Announcement Logic with Boolean Definitions}

\newcommand{\orcid}[1]{\\{\href{https://orcid.org/#1}{\footnotesize\tt #1}}}

\author{
Malvin Gattinger
\institute{Bernoulli Institute, University of Groningen}
\orcid{0000-0002-2498-5073}
\email{malvin@w4eg.eu}
\and
Yanjing Wang\footnote{Corresponding author}
\institute{Department of Philosophy, Peking University}
\orcid{0000-0002-9499-416X}
\email{y.wang@pku.edu.cn}
}

\def\titlerunning{Public Announcement Logic with Boolean Definitions}
\def\authorrunning{Malvin Gattinger \& Yanjing Wang}

\begin{document}
\maketitle

\begin{abstract}
In standard epistemic logic, knowing that $p$ is the same as knowing that $p$ is true, but it does not say anything about understanding $p$ or knowing its meaning.
In this paper, we present a conservative extension of Public Announcement Logic (PAL) in which agents have knowledge or belief about both the truth values and the meanings of propositions.
We give a complete axiomatization of \emph{PAL with Boolean Definitions} and discuss various examples.
An agent may understand a proposition without knowing its truth value or the other way round.
Moreover, multiple agents can agree on something without agreeing on its meaning and vice versa.
\end{abstract}


\section{Introduction}
\textit{Konnyaku Mondo} (jelly dialogue) is a story from the  traditional
Japanese comic storytelling in the \textit{rakugo} form.
Quoting~\cite{Mondo04}, the story goes like the following:
\begin{quote}
\emph{There was a temple where no monks were living any longer. A devil's
tongue jelly maker, named Rokubei, lived next door. He moved into the temple and
started pretending to be a monk. One day, a traveling Zen Buddhist monk passed
by and challenged Rokubei to a debate on Buddhism, Rokubei had no knowledge on
Buddhism and was not able to have a debate. He tried to refuse, but he could not
escape and finally agreed. The Buddhist dialogue started but Rokubei didn't know
how to perform and therefore, he kept silent. The Buddhist monk tried to
communicate to Rokubei in many ways. After some time, Rokubei started responding
with gestures to the body movements the monk made. The monk took this as a style
of dialogue and tried to answer in gestures, too. They exchanged gestures, and
after some time, the monk told Rokubei, ``your thoughts are profound and mine
are of no comparison. I am very sorry to have bothered you'', After saying this,
he left the temple.}
\end{quote}
In fact, the monk thought ``the master'' (Rokubei) had expressed deep Buddhist
thoughts by his gestures, but Rokubei had never learned any Buddhist thoughts.
Rather, from some stage on, he thought the monk was talking badly about his
jelly with those gestures, thus he gave the monk a lesson by some angry moves,
and apparently defeated the monk.

The intriguing nature of the story is that, as remarked in~\cite{Mondo04}, it seems the proposition \textit{Rokubei has defeated the monk} is common knowledge between the two, but it is based on mutual misunderstanding.
The two actually have completely different understanding for the commonly agreed ``defeat of the monk''.
Such mutual misunderstanding also happens a lot in everyday life communications, even in academic exchanges when people ``agree'' to the same thing due to different interpretations or definitions of the same concept.
See~\cite{Mondo04} for excellent (and entertaining) interactive discussions about such Konnyaku Mondo phenomena in Game theory.

Mutual misunderstanding is not always harmful.
To postpone immediate conflicts and achieve some consensus, it is sometimes even intended to allow respective interpretations of the same proposition, which happens in diplomatic scenarios.
For example, two brothers may disagree about who represents officially their father, but they may reach the temporary consensus that there is one and only one successor in order to avoid immediate conflicts.

Philosophically, if we require that knowledge should be at least properly
justified as Gettier's examples suggest~\cite{Gettier}, we can hardly say both
Rokubei and the monk ``know'' that Rokubei has defeated the monk, since the same
proposition is justified by two different reasons by different parties. The
tricky thing here is that perhaps there is no single ``real'' justification for
the defeat. More crucially, it is debatable in this particular case, whether
there is a fixed ``real'' meaning of the proposition that Rokubei has defeated
the monk.

As logicians, the story makes us think about how to represent such situations in the framework of epistemic logic, where $K_i p$ expresses that agent $i$ knows that $p$.
According to the standard Kripke semantics, $K_i p$
merely means that agent $i$ is certain that $p$ is true, and there is nothing
about the meaning of $p$ in the semantics. For example, suppose you do not
understand Chinese, but someone said a Chinese sentence $p$ and guaranteed you
its truth, then it seems you indeed know the \textit{truth value} of $p$, but
without knowing its \textit{meaning}. Now suppose the speaker then tells you
that by $p$, he actually meant $r\land q$ in your own language, then you know
both the meaning of $p$ and its truth value (and the truth values of $q$ and
$r$). Note that, even when $p$ is uttered in a language that you know, it still
can have a different meaning than the surface one, e.g, when a Chinese says ``we
will think about it later'' when asked a yes-no question, it often means ``no''.
It is crucial to see that knowing that \textit{$p$ means $\phi$} is different from knowing that $p \lra \phi$, where the latter is again simply about the truth value.
For example, knowing that $p \lra \top$ does not imply that \textit{$p$ means $\top$}.
We need new techniques to handle knowledge of meanings in epistemic logic.

More generally, knowing the value of something does not imply understanding what information it carries.
For example, knowing the ID number of a Chinese person (e.g., 110105198002290022) does not tell you much, unless you know that the first 6 digits encode the residence (Chaoyang District in Beijing), the next 8 digits encode the birth date (29th of February 1980) and so on.
Clearly, the interpretation of the structure of the ID is important to your understanding.
You may also only know part of the meaning of the message and there are different ``layers'' of your understanding.
You may or may not know that the gender of the person is given by the parity of the second to last digit (in this case even for female).
Merely knowing that the first 6 digits code the residence does not tell you the exact city where this person is registered, you may need to know further that the fist three digits code the city and 110 is the code for Beijing.
Therefore, by limited knowledge of the structure, you may only get part of the meaning.

Coming back to the propositional setting in this paper, for a sentence $\phi$,
you might also just understand some part of it, e.g., knowing what $q$ means in
$p \lor q$ but only understand $p$ as $\neg r$ for some incomprehensible $r$.
Now, if others elaborate that $r$ means $q\land o$, then you have a deeper
understanding of $p\lor q$. We may also enhance our understanding by matching
the structure of the proposition --- if someone utters $(p\land p') \lor p''$
and later explains that it means $q \lor (r\land r')$, then we know $q$ means
$p\land p'$, and $p''$ means $(r\land r')$.
The technical goal of this paper is to formally flesh out such reasoning patterns with a minimal extension of public announcement logic~\cite{Plaza89:lopc,Plaza2007:LoPC}.
The basic idea is to treat incomprehensible propositions as atomic propositions with \textit{boolean definitions} based on the basic atomic propositions.

Concerning related work, the distinction between knowing the value and knowing the meaning is crucial in cryptographic message passing.
The goal of encryption schemes is to hide the meaning of a message, even if the transmitted ciphertext is known~\cite{DeciSP,CohenD07}.
Another related recent attempt~\cite{Ding2016:FDep}, based on a logic of knowing value~\cite{WangFan2013KvPAL,vEGW2017:KvPIL,Baltag2016:KVV}, introduces a functional dependency operator to express that the agent knows a function which can explain the dependency between variables $c$ and $d$.
However, as the author of~\cite{Ding2016:FDep} also remarked, it is not enough to capture the meaning of variables.
In~\cite{LiuW13}, the meaning of an utterance by an agent depends on the \emph{type} of the agent, which is a function mapping the uttered proposition to its actual meaning.
Similarly, protocols can also give meaning to communicative actions, as demonstrated in~\cite{BarSel97:infoflow,Wang11,DGVW14}, where the meaning of an action is defined by the corresponding precondition of the protocol regarding this.
We discuss other related work in the conclusion.

In the rest of this paper, we first layout the language and semantics of our
logic in Section~\ref{sec.ls}, and then axiomatize it in Section~\ref{sec.ax},
before concluding with ideas for future work.

\section{PAL with Boolean Definitions}\label{sec.ls}

\subsection{Language, Models, Semantics}

Throughout the paper our languages and models are parameterized by a set of proposition letters $\mathsf{Prop}$ and a finite set of agents $I$.
The language we study consists of two parts: a purely boolean layer for which our models will also provide definitions, and on top of that a version of Public Announcement Logic (PAL) as in~\cite{Plaza89:lopc,DitHoekKooi2007:del}.

\begin{definition}[Language]
The boolean language $\Lng_B$ is defined by
\[ P  ::=  p  \mid  \lnot P  \mid  (P \land P) \]
where $p \in \Prop$, a countable set of propositional letters.
We already note that the parentheses are essential for pattern matching as we will see later.

The full language $\Lng$ is given by
\[ \phi  ::=  P  \mid  P \equiv P  \mid  \lnot \phi  \mid  (\phi \land \phi)  \mid  \Box_i \phi  \mid  [\phi]\phi \]
where $P \in \Lng_B$ and $i \in I$.
We write $Q \not\equiv P$ for $\lnot(Q \equiv P)$
and use the standard abbreviations for other boolean operators
  $\phi \lor \psi := \lnot (\lnot \phi \land \lnot \psi)$,
  $\phi \to \psi := \lnot \phi \lor \psi$ and
  $\phi \leftrightarrow \psi := (\phi \to \psi) \land (\psi \to \phi)$.
\end{definition}

For modalities we use the notation $\Box_i$ and avoid the symbol $K_i$ which would suggest knowledge.
This is because we will not assume additional modal axioms besides $K$.
Our framework can easily be adapted to multi-agent $KD45$, $S5$ and other logics.

We read $P_1 \equiv P_2$ as ``$P_1$ and $P_2$ have the same meaning'' or ``$P_1$ and $P_2$ are equivalent by definition''.
For example, if $p$ means $r \lor r'$ and $q$ means $\neg o$, then our semantics below will also imply that $p \land \neg o$ has the same meaning as $((r \lor r') \land q)$.

We avoid reading the $\equiv$ operator as ``\dots is defined as \dots'' which would suggest a directedness and uniqueness on the right side.
In fact, while our models will contain unique definitions for which one might write $:=$, we do not refer to \emph{the} definition.
The formal language can only express that certain propositions are \emph{equivalent by definition}.

A model in our framework is a Kripke model with a local \emph{definition function} $\DF_w$ on each world $w$, which assigns to each $p\in \Prop$ a definition as its (most thorough) meaning.
To stay as general as possible, we \emph{do not} assume any frame property in this paper.
Intuitively, if a proposition letter $p$ is assigned itself as the definition, this means that $p$ is \emph{self-evident} or truly basic.
Based on those definitions, we can then ``unravel'' each boolean formula in $\Lng_B$ to obtain its meaning by recursively applying $\DF_w$.
Further constraints are imposed to avoid circularity and make sure each non-self-evident proposition is assigned a definition using self-evident propositions.

\begin{definition}[Models]\label{def:models}
A \emph{premodel} $\M$ is a tuple $(W,R,V,\DF)$ where
  $W$ is a non-empty set of worlds,
  $R_i \subseteq W \times W$ is a relation for each agent $i$,
  $V \colon W \to \Prop \to \{0,1\}$ is a valuation function
  and $\DF \colon W \to \Prop \to \Lng_B$ is a definition function.

We write $V_w$ for the valuation at $w$ and lift it to $\Lng_B$ as usual by the
standard boolean semantics included in Definition~\ref{def:semantics} below.
Similarly, we lift the definition function $\DF_w$ at $w$ from $\Prop$ to $\Lng_B$
using definitions from the premodel for atoms and recursing over $\lnot$ and $\land$.
Formally, let $\df_w : \Lng_B  \to \Lng_B$ be defined by:
\[ \begin{array}{rcl}
  \df_w(p)             & := & \DF_w(p)\\
  \df_w(\lnot P)       & := & \lnot \df_w(P)\\
  \df_w((P_1 \land P_2)) & := & (\df_w(P_1) \land \df_w(P_2))\\
\end{array} \]
Intuitively, $\df_w(P)$ is obtained by replacing all the propositions letters in $P$ by their definitions.

A premodel is a \emph{model} iff we have for all worlds $w$:
\begin{itemize}
  \item For all $P,Q \in \Lng_B$:
    If $\df_w(P)=\df_w(Q)$, then $V_w(P) = V_w(Q)$.
  \item For all $p,q \in \Prop$:
    If $p$ is in $\DF_w(q)$, then $\DF_w(p) = p$.
\end{itemize}
\end{definition}

To connect definitions and truthvalues the first model constraint demands that whatever is definitionally equivalent is also assigned the same truth value.
We note that only demanding this condition for atomic propositions does not suffice for our purposes.

The second model constraint ensures well-foundedness:
Models never contain circular definitions like $p := (p \land q)$.
Moreover, they also do not contain chains of definitions that would imply such a
definition if they were unraveled, for example $p := r$ and $r := (p \land q)$.
While some of these might actually make sense as fixpoints or infinite
conjunctions, for now we do not allow them in our framework.

\begin{definition}[Semantics]\label{def:semantics}
We interpret $\Lng$ on models as follows.
\[ \begin{array}{lcl}
\M,w \vDash p               & \iff & V_w(p) = 1 \\
\M,w \vDash P_1 \equiv P_2  & \iff & \df_w(P_1) = \df_w(P_2) \\
\M,w \vDash \lnot \phi      & \iff & \text{not } \M,w \vDash \phi \\
\M,w \vDash (\phi \land \psi) & \iff & \M,w \vDash \phi \text{ and } \M,w \vDash \psi \\
\M,w \vDash \Box_i \phi     & \iff & \text{for all }v : w R_i v \text{ implies } \M,v \vDash \phi \\
\M,w \vDash [\phi] \psi     & \iff & \M,w \vDash \phi \text{ implies } \M^\phi,w \vDash \psi \\
\end{array} \]
where $\M^\phi$ is the restriction of $\M$ to the new set of worlds
$\{ v \in W \mid \M,v \vDash \phi \}$.
\end{definition}

In the second condition, the $=$ symbol on the right side is \emph{syntactic} equality within $\Lng_B$.
As we mentioned, parentheses in $\equiv$ formulas matter, e.g., $(p \land (q \land r)) \not\equiv ((p\land q) \land r)$ is valid.
In contrast $(p\land (q \land r)) \lra ((p \land q) \land r)$ is clearly valid, where we can omit the parentheses.
Note that all clauses besides the one for $\equiv$ are standard PAL as in~\cite{Plaza89:lopc,DitHoekKooi2007:del}.
In particular, the result of announcements is defined as usual, preserving the valuation and now also the local definitions at each world.

\subsection{Examples}

To illustrate our semantics and to show that it can describe various different scenarios, we now give some examples.

\begin{example}[Knowing without Understanding]\label{ex:knowing-without-understanding}
As mentioned in the introduction, you can know that something is true without knowing what it means and without knowing that its meaning is true.
Figure~\ref{fig:knowing-without-understanding} shows such a model.
We use undirected edges to indicate an equivalence relation and we omit the self-evident definitions for $q$ and $r$ in all three worlds.
At the actual world in the middle we have $\Box_i p \land (p \equiv q) \land \lnot \Box_i (p \equiv q)\land \neg \Box_i q$.
Moreover, even if $p \leftrightarrow q$ were announced, only the right world would be removed.
The agent would then know that $p \leftrightarrow q$ but still not know the stronger $p \equiv q$.
\end{example}

\begin{figure}[H]
\begin{minipage}[t]{0.5\textwidth}
\centering
\begin{tikzpicture}[>=latex,line join=bevel, node distance=3cm]
  \node (0) [draw,double] {$\begin{array}{l}
    p \\
    q \\
    \lnot r \\
    \hline
    p:=q
    \end{array}$};
  \node (1) [draw,right of=0] {$\begin{array}{l}
    p\\
    \lnot q\\
    r \\
    \hline
    p:=r
    \end{array}$};
  \draw (0) -- node[outer sep=1mm,above]{$i$} (1);
  \node (2) [draw,left of=0] {$\begin{array}{l}
    p\\
    q\\
    r \\
    \hline
    p:=r
    \end{array}$};
  \draw (0) -- node[outer sep=1mm,above]{$i$} (2);
\end{tikzpicture}
\caption{Knowing without Understanding}\label{fig:knowing-without-understanding}
\end{minipage}\begin{minipage}[t]{0.5\textwidth}
\centering
\begin{tikzpicture}[>=latex,line join=bevel, node distance=3cm]
  \node (0) [draw,double] {$\begin{array}{l}
    p \\
    q \\
    \hline
    p:=q
    \end{array}$};
  \node (1) [draw,right of=0] {$\begin{array}{l}
    \lnot p\\
    \lnot q\\
    \hline
    p:=q
    \end{array}$};
  \draw (0) -- node[outer sep=1mm,above]{$i$} (1);
\end{tikzpicture}
\caption{Understanding without Knowing}\label{fig:understanding-without-knowing}
\end{minipage}
\end{figure}

\begin{example}[Understanding without Knowing]\label{ex:understanding-without-knowing}
Conversely, an agent can know the meaning of a proposition but still not know whether it is true.
Both worlds in the model shown in Figure~\ref{fig:understanding-without-knowing} satisfy
$\Box_i (p \equiv q) \land \Box_i (p \leftrightarrow q) \land \lnot \Box_i p$.
\end{example}

\begin{example}[Understanding different parts]
Two agents can also have different partial knowledge of the meaning of some proposition.
The middle world of the model in Figure~\ref{fig:understanding-different-parts} satisfies
these three formulas:
\[ \Box_a(p \equiv (q \land r)) \land \Box_b(p \equiv (q \land r)) \]
\[ \Box_b(p \equiv (\lnot q_1 \land r)) \land \lnot \Box_a(p \equiv (\lnot q_1 \land r)) \]
\[ \Box_a(p \equiv (q \land \lnot r_1)) \land \lnot \Box_b(p \equiv (q \land \lnot r_1)) \]
Note that if the agents would combine their knowledge by announcing both partial meanings, they would arrive at the most thorough meaning, which is more informative than what either of them knows at the moment.
Formally, we have $[r \equiv \lnot r_1][q \equiv \lnot q_1] (\land_{i=a,b} \Box_i (p \equiv (\lnot q_1 \land (\lnot r_1 ))))$.

\begin{figure}[H]
\centering
\begin{tikzpicture}[>=latex,line join=bevel, node distance=4cm]
  \node (0) [draw] {$\begin{array}{l}
    \lnot p, \lnot q, r\\
    \lnot q_1, \lnot r_1\\
    \hline
    p := q \land \lnot r_1\\
    q := q\\
    r := \lnot r_1 \\
    \end{array}$};
  \node (1) [draw,right of=0,double] {$\begin{array}{l}
    p, q, r\\
    q_1, r_1\\
    \hline
    p := \lnot q_1 \land \lnot r_1\\
    q := \lnot q_1 \\
    r := \lnot r_1 \\
    \end{array}$};
  \draw (1) -- node[outer sep=1mm,above]{$a$} (0);
  \node (2) [draw,right of=1] {$\begin{array}{l}
    \lnot p, q, \lnot r\\
    \lnot q_1, \lnot r_1\\
    \hline
    p := \lnot q_1 \land r\\
    q := \lnot q_1\\
    r := r\\
    \end{array}$};
  \draw (1) -- node[outer sep=1mm,above]{$b$} (2);
\end{tikzpicture}
\caption{Understanding different parts}\label{fig:understanding-different-parts}
\end{figure}
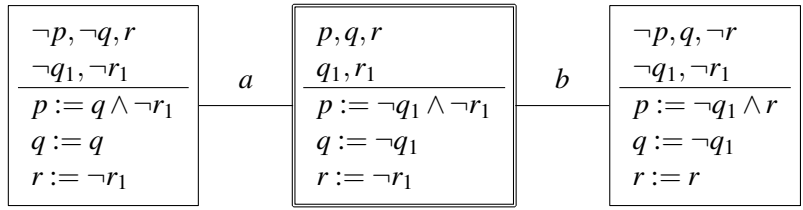

\begin{figure}[H]
\centering
\begin{tikzpicture}[>=latex,line join=bevel, node distance=3cm]
  \node (0) [draw] {$\begin{array}{l}
    p\\
    \lnot q\\
    r\\
    \hline
    p:=r
    \end{array}$};
  \node (1) [draw,right of=0,double] {$\begin{array}{l}
    p\\
    \lnot q\\
    \lnot r\\
    \hline
    p:=p
    \end{array}$};
  \draw[->] (1) -> node[outer sep=1mm,above]{$i$} (0);
  \node (2) [draw,right of=1] {$\begin{array}{l}
    p\\
    q\\
    \lnot r\\
    \hline
    p:=q\\
    \end{array}$};
  \draw[->] (1) -- node[outer sep=1mm,above]{$j$} (2);
\end{tikzpicture}
\caption{Consensus with misunderstanding}\label{fig:consens-misunder}
\end{figure}
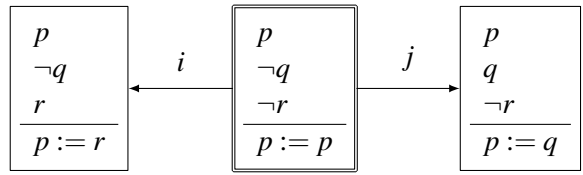
\end{example}

\begin{example}[Consensus with misunderstanding]
As in the \textit{Konnyaku Mondo} example, two agents can agree on something but actually have different beliefs about what it means.
In the model from Figure~\ref{fig:consens-misunder}, at the middle world, we have
  $p \land \Box_i p \land \Box_j p$
but also
  $\Box_i ((p \equiv r) \land r \land \lnot q)$
and
  $\Box_j ((p \equiv q) \land q \land \lnot r)$.
\end{example}

\section{Axiomatization}\label{sec.ax}

Note that the addition of $\equiv$ does not invalidate the standard reduction
axioms of PAL\@. We can rather think of $\phi \equiv \psi$ as a new atomic
proposition which is evaluated purely locally.

\begin{definition}\label{def:circularForm}
A formula $P \equiv Q$ is a \emph{circular} formula iff there is an atomic $p$ such that either
  (i) $P = p$ and $Q \neq p$ and $p$ occurs in $Q$
  or
  (ii) vice versa.
\end{definition}

For example the formulas $p \equiv (p \land q)$ and $\lnot q \equiv q$ are both circular, but $p \equiv p$ is not circular.

\begin{definition}[Axioms and Rules]\label{def:axioms}
We define the following proof system $\mathbb{SPALD}$ for $\Lng$.
We call the system without the PAL reduction axioms $\mathbb{SMLD}$.

\noindent\paragraph{Axiom Schemes}
    \begin{itemize}
      \item All propositional tautologies.
      \item The $K$ axiom:
        $\Box_i(\phi \to \psi) \to (\Box_i \phi \to \Box_i \psi)$
      \item PAL reduction axioms:
        \begin{itemize}
          \item $[\phi] p \leftrightarrow (\phi \to p)$
          \item $[\phi] (P \equiv Q) \leftrightarrow (\phi \to (P \equiv Q))$
          \item $[\phi] \lnot \psi \leftrightarrow (\phi \rightarrow \lnot[\phi]\psi)$
          \item $[\phi] (\psi \land \theta) \leftrightarrow ([\phi]\psi \land [\phi]\theta)$
          \item $[\phi] \Box_i \psi \leftrightarrow (\phi \rightarrow \Box_i(\phi \rightarrow [\phi]\psi))$
          \item $[\phi] [\psi] \xi \leftrightarrow [\phi \land [\phi]\psi]\xi$
        \end{itemize}
      \item Definition axioms
        \begin{itemize}
          \item Reflexivity:
            $(P \equiv P)$
          \item Symmetry:
            $(P \equiv Q) \to (Q \equiv P)$
          \item Transitivity:
            $((P \equiv Q) \land (Q \equiv R)) \to (P \equiv R)$
          \item Equivalence:
            $(P \equiv Q) \to (P \leftrightarrow Q)$
          \item Atomic \emph{occurrence} substitution:
            $((p \equiv Q) \land (R \equiv S)) \to (R \equiv [k \col p \mapsto Q]S)$
            where $k \col p$ denotes the $k$th occurrence of $p$ in $S$.
          \item Pattern $\lnot$:
            $(\lnot P \equiv \lnot Q) \leftrightarrow (P \equiv Q)$
          \item Pattern $\land$:
            $((P \land Q) \equiv (R \land S)) \leftrightarrow ((P \equiv R) \land (Q \equiv S))$
          \item Pattern mismatch: $\lnot P \not\equiv (Q \land R)$
          \item Non-circularity: $p \not\equiv P$ where $p$ occurs in $P$ but $P \neq p$
        \end{itemize}
    \end{itemize}
 \noindent\paragraph{Rules}
    \begin{itemize}
      \item Modus Ponens: from $\phi$ and $\phi \to \psi$ infer $\vdash \psi$.
      \item Necessitation for $\Box_i$: from $\vdash \phi$ infer  $\Box_i \phi$.
    \end{itemize}
\end{definition}

Somewhat non-standard in our axiom system is the usage of occurrence substitutions, in contrast to standard substitutions which usually replace all occurrences of a given atom.
We also use the notation $[k \col p \mapsto Q]$ in proofs in the next section and it will become more clear there.

We note that our logic does not allow replacement of equivalents in $\equiv$ formulas, hence this is not an admissible proof rule.
For example, $p \leftrightarrow (p \land p)$ is valid and so is $p \equiv p$, but $p \equiv (p \land p)$ is not valid and in fact a contradiction, by the non-circularity axiom.
Finally, we note that necessitation for $[\phi]$ is an admissible rule~\cite{WC12}.

\begin{theorem}
The axiom system $\mathbb{SPALD}$ from Definition~\ref{def:axioms} is sound for the semantics from Definition~\ref{def:semantics}.
\end{theorem}

\section{Completeness}

To show the completeness of our axiomatization, we can first show the completeness of $\SMLD$ for announcement-free formulas.
Then by using the reduction axioms we can obtain the completeness of $\SPALD$ in the usual way~\cite{Plaza89:lopc}.
In the following, we write  $\vdash \phi$ if $\phi$ is provable in $\SMLD$.

Before we go on, note that our axioms enforce non-circularity but not well-foundedness.
That is, the set
  $\{ p_i \equiv p_{i+1}\land p_{i+2} \mid i \in \mathbb{N} \}$
is consistent according to the proof system above, but it does not have a model, since you cannot give definitions to $p_i$ using some self-evident atoms.
However, any finite subset of this set has a model, hence our logic is not compact and we cannot show strong completeness.

Fortunately, we can still show completeness because any finite set of formulas only uses finitely many atomic propositions.

For the rest of this section, we fix a finite vocabulary $\Prop$.

\begin{example}\label{ex:growNonCirc}
Consider the following three formulas and an infinite sequence of consequences:
\[ \begin{array}{ll}
  p \equiv q \land r\\
  q \equiv p \land r\\
  s \equiv p\\
  \midrule
  s \equiv q \land r\\
  s \equiv (p \land r) \land r\\
  s \equiv ((q \land r) \land r) \land r\\
  \vdots\\
\end{array} \]
Note that none of these formulas alone is circular.
However, it is easy to spot another consequence $p \equiv (p \land r) \land r$ which is circular.
Hence the original set of three formulas is inconsistent.
\end{example}

The main idea for our completeness proof is that Example~\ref{ex:growNonCirc} is not an exception:
Whenever a set of equivalent definitions is infinite we can systematically derive a circular formula.
Before stating our central Lemma~\ref{lem:finiteDefSet} we need a few more definitions.

\begin{definition}
  For each boolean formula $P \in \Lng_B$ we define its \emph{length} $l(P)$ as follows:
  \[ \begin{array}{lcl}
    l(p)         & := & 1 \\
    l(\neg P)    & := & l(P) + 1 \\
    l((P \land Q)) & := & l(P) + l(Q) + 3
  \end{array} \]
  Additionally, we define its \emph{vocabulary} $v(P)$ as follows:
  \[ \begin{array}{lcl}
    v(p)         & := & \{ p \} \\
    v(\neg P)    & := & v(P) \\
    v((P \land Q)) & := & v(P) \cup v(Q)
  \end{array} \]
\end{definition}

As an example, $l(\lnot p) = 2$ and $l(p\land (p\land q)) = 9$.
Note that the parentheses also count.

\begin{definition}
Given a maximally consistent set $\Gamma \subseteq \Lng(\Prop)$, we define a relation over $\Lng_B(\Prop)$ by $P \equiv_\Gamma Q \ :\iff \ (P \equiv Q) \in \Gamma$.
Per relevant axioms in Definition~\ref{def:axioms} this is an equivalence relation.
For each $P \in \Lng_B(\Prop)$ we denote its $\equiv_\Gamma$-equivalence class by
\[ {[P]}_\Gamma = \{ Q \in \Lng_B(\Prop) \mid (P \equiv Q) \in \Gamma \} \]
and call it the \emph{set of $\Gamma$-definitions} of $P$.
\end{definition}

\begin{lemma}\label{lem:finiteDefSet}
For each maximally consistent set $\Gamma \subseteq \Lng(\Prop)$ and each atomic proposition $p \in \Prop$, the set ${[p]}_\Gamma$ is finite.
\end{lemma}

To prove Lemma~\ref{lem:finiteDefSet}, we first need some definitions and notation.
We fix an enumeration of $\Prop$, say alphabetically.
This induces a lexicographic ordering $<$ over $\Lng_B$.
For example, we have $p < q$, therefore also $\lnot p < \lnot q$ and similarly for conjunctions.

\begin{definition}\label{def:mrg}
Let $\mrg \colon \Lng_B\times \Lng_B \to \Lng_B$ be defined as follows:
\[ \begin{array}{lcll}
\mrg(p,q) & := & \text{if } p < q \text{ then } p \text{ else } q \\
\mrg(p,Q)\ \text{ when } Q\not=q & := & Q \\
\mrg(P,q) \ \text{ when }P\not=p & := & P \\
\mrg((P \land Q), (R \land S)) & := & (\mrg(P,R) \land \mrg(Q,S)) \\
\mrg(\lnot P, \lnot R) & := & \lnot \mrg(P,R) \\[0.5em]
\mrg(\lnot P, (Q \land R)) & := & \text{undefined} \\
\mrg((Q \land R), \lnot P) & := & \text{undefined} \\
\end{array} \]
\end{definition}

By definition, $\mrg$ is symmetric: $\mrg(P,Q)=\mrg(Q,P)$.

\begin{example}
  We have $\mrg((p \land (q \land r)), (\lnot s \land t)) = (\lnot s \land (q \land r))$.
\end{example}
It is also easy to see that, viewed as  term-rewriting rules, $\mrg$ always terminates, since the recursive clauses reduce the complexity of the formulas.
\begin{definition}
  Let $[ k \col p \mapsto P] Q$ denote the result of replacing the $k$-th occurrence of $p$ in $Q$ with $P$.
  For multiple such \emph{occurrence substitutions}, let $[k \col p \mapsto P \ \oplus \ k' \col p' \mapsto R] Q$ denote their simultaneous application to $Q$.
\end{definition}
\begin{example}
  We have $[ 2 \col p \mapsto (q \land r)] (p \land p) \ = (\ p \land (q \land r))$.
  As an example of two simultaneous occurrence substitutions, we have
  $ [ 2 \col p \mapsto (q \land r) \ \oplus \ 1 \col q \mapsto \lnot r ]
    ((p \land p) \land q)
    \ = \
    ((p \land (q \land r)) \land \lnot r) $.
\end{example}

\begin{lemma}\label{lem:mergeNicepre}
  For all $P, Q$ such that $P\equiv Q\in \Gamma$ we have:
  \begin{enumerate}[(i)]
  \item $\mrg(P,Q)\equiv P \in \Gamma$
  \item There are
    indices $k_1,\ldots,k_n \in \mathbb{N}$,
    atoms $p_1,\ldots,p_n \in \Prop$ and
    formulas $R_1,\ldots,R_n \in \Lng_B$ for some fixed $n \geq 0$
    such that we have
    $\mrg(P,Q) = [k_n \col p_n \mapsto R_n \ \oplus \ldots \oplus \ k_1 \col p_1 \mapsto R_1] P$ and $p_i
    \equiv R_i\in \Gamma$ for all $i\leq n$.
    (Note that $n = 0$ iff $\mrg(P,Q)=P$.)
  \end{enumerate}
\end{lemma}
\begin{proof}
Since $\Gamma$ is consistent, if $P\equiv Q\in \Gamma$ then $\mrg(P,Q)$ is always defined, based on the axioms of pattern mismatch.

For (i) we do induction on the structure of $P$.

Suppose $P=p$, then $\mrg(P, Q)$ is $p$ or $Q$.
Then it is straightforward that $\mrg(P, Q)\equiv P \in \Gamma$, since $P\equiv Q\in \Gamma$.

Suppose $P = \neg P'$.
Since $P\equiv Q\in \Gamma$, the shape of $Q$ is either $q$ or $\neg Q'$ due to the pattern mismatch axiom and the fact that $\Gamma$ is consistent.
The first case reduces to the above case due to the axiom of symmetry and the fact that $\mrg$ is symmetric.
For the second case, by pattern inference we have $P'\equiv Q'\in \Gamma$.
Now by induction hypothesis, $\mrg(P', Q')\equiv P'\in\Gamma$.
Therefore $\neg\mrg(P',Q')\equiv \neg P'\in\Gamma$ by pattern inference axiom, namely $\mrg(P,Q)\equiv P\in \Gamma$.

The case of $P=(P_1\land P_2)$ is similar.

For (ii), we also do induction on the structure of $P$.

Suppose $P=p$, then $\mrg(P, Q)$ is $p$ or $Q$.
In the first case we just need a trivial substitution $[1 \col p \mapsto p]$ since $p \equiv p \in \Gamma$.
In the second case we take $[1 \col p \mapsto Q]$ since $p\equiv Q\in \Gamma$.

Suppose $P=\neg P'$.
As in the proof of (i) we can show that $Q$ is in the shape of either $q$ or $\neg Q'$.
For the first case, $\mrg(P, q)= P=\neg P'$ by definition of $\mrg$.
Then we can take the trivial substitution to prove the claim.
In the second case, $\mrg(P, Q)=\neg \mrg(P', Q')$.
Since $P\equiv Q\in \Gamma$, $P'\equiv Q'\in\Gamma$ by pattern inference.
Now by induction hypothesis, there are substitutions to turn $P'$ into $\mrg(P', Q')$.
The same substitutions can also turn $P=\neg P'$ into $\mrg(P, Q)=\neg\mrg(P',Q')$.

Again the case of $P=(P_1 \land P_2)$ is similar.
\end{proof}

\begin{lemma}\label{lem:mergeNice}
  For all $p \in \Prop$ and all $P, Q \in {[p]}_\Gamma$ we have:
  \begin{enumerate}[(i)]
  \item $\mrg(P,Q) \in {[p]}_\Gamma$
  \item $l(\mrg(P,Q)) \geq \max(l(P), l(Q))$
  \item There are
    indices $k_1,\ldots,k_n \in \mathbb{N}$,
    atoms $p_1,\ldots,p_n \in \Prop$ and
    formulas $R_1,\ldots,R_n \in \Lng_B$ for some fixed $n \geq 0$
    such that we have
    $\mrg(P,Q) = [k_n \col p_n \mapsto R_n \ \oplus \ldots \oplus \ k_1 \col p_1 \mapsto R_1] P$ and $p_i
    \equiv R_i\in \Gamma$ for all $i\leq n$.
  \end{enumerate}
\end{lemma}
\begin{proof}
For (i): Suppose $P, Q \in {[p]}_\Gamma$.
Then $P \equiv Q \in \Gamma$ and according to part (i) of Lemma~\ref{lem:mergeNicepre} we have $\mrg(P,Q) \equiv P \in \Gamma$.
Since $P\in {[p]}_\Gamma$, thus by transitivity axiom, $\mrg(P,Q)\in{[p]}_\Gamma$.

For (ii): a simple induction on the clauses of $\mrg$ suffices.

Finally, (iii) is a special case of part $(ii)$ in Lemma~\ref{lem:mergeNicepre} since $P,Q \in {[p]}_\Gamma$ implies $P \equiv Q \in \Gamma$.
\end{proof}

Intuitively, Lemma~\ref{lem:mergeNice} says that the result of merging two $\Gamma$-definitions of $p$
(i) is also a $\Gamma$-definition of $p$,
(ii) is at least as long as the longest given formula, and
(iii) can be reached by replacing atom occurrences step by step.

In fact, the occurrence substitutions given by (iii) are unique up to enumeration and trivial substitutions of the form $[k \col p \mapsto p]$.
Moreover, the pattern matching axioms imply that $R_i \in {[p_i]}_\Gamma$.

\begin{proof}[Proof of Lemma~\ref{lem:finiteDefSet}]
Suppose ${[p]}_\Gamma$ is infinite.

Then in particular the length of formulas in ${[p]}_\Gamma$ is unbounded, because $\Prop$ is finite.
Hence there must be an infinite chain $P_0 = p, P_1, P_2, \ldots$ in ${[p]}_\Gamma$ which is unbounded in length.
Note that there might be big ``jumps'' in length and in general the formulas will not be related systematically.

To deduce a circular formula and thus a contradiction, we now define a second chain $Q_0,Q_1,\ldots$ using $\mrg$.
Let $Q_0 := P_0 = p$ and for all $k \geq q$ let $Q_k := \mrg(Q_{k-1},P_k)$.
From Lemma~\ref{lem:mergeNice} we now get that
(i) the chain of $Q_i$s is also a chain in ${[p]}_\Gamma$,
(ii) $l(Q_i) \geq l(P_i)$ and thus this chain is also unbounded in length, and
(iii) for each step from $Q_i$ to $Q_{i+1}$ there are finitely many atoms that are replaced with longer formulas.

Let us list such a sequence of substitutions as:
\[ Q_0 \xrightarrow{\overline{k\col p\mapsto R}} Q_1 \xrightarrow{\overline{m\col q\mapsto T}} Q_2 \ \cdots \]
where $\overline{k\col p\mapsto R}$ denotes the simultaneous substitution $[k_{n} \col p_n \mapsto R_n \ \oplus \ldots \oplus \ k_1 \col p_1 \mapsto R_1]$ for some $n$.
We can denote each single substitution as ${[k:p\mapsto R]}^i$ where the superscript $i$ means the substitution happens at $Q_i$.

Now let $(N,\rightsquigarrow)$ be the graph where
  $N$ is the set of all occurrence substitutions in the $Q$ chain (indexed with superscripts) and
  there is an edge ${[m \col p \mapsto R]}^i \rightsquigarrow {[n \col q \mapsto S]}^j$ iff $i<j$ and $n \col q$ is an occurrence within $R$ of $Q_{i+1}$.
Then $(N,\rightsquigarrow)$ is a tree with the first substitution ${[1 \col p \mapsto P_1]}^0$ as its root.
Intuitively, we have an edge from one substitution to another iff the second ``happens within'' the result of the first.
This tree of substitutions is illustrated in Figures~\ref{fig:subTreeLine} and~\ref{fig:subTreeExample} below as examples.

The $Q_i$ chain of formulas is infinite, hence there are infinitely many substitutions and $N$ is infinite.
However, each occurrence of an atomic proposition can be replaced with a longer formula only once.
Hence the number of children of each node in $(N,\rightsquigarrow)$ is bounded by the length of the formula.
Formally, each node ${[k \col p \mapsto R]}^i$ can only have as many children as there are occurrences of atoms in $R$.

Together, $(N,\rightsquigarrow)$ is an infinite but also finitely branching tree.
By Kőnig's Lemma~\cite{Koenig1927,Fran1997:OrigKoenig} there must be an infinite branch.
In particular, there must be a branch longer than $|\Prop|$ and there must be two substitutions of the same atom along this branch.
We now use this branch to derive a circular formula in $\Gamma$.

Let $R_i$ be the sequence of formulas starting with $R_0 := p$ and then applying the substitutions from the branch.
For each $i$, let ${[k_i \col q_i \mapsto S_i]}^{\iota_i}$ be the substitution at node $i$ happening at $Q_{\iota_i}$.
Note that the branch need not have a node at every level, hence $\iota_i \geq i$ but not necessarily $\iota_i = i$.

Because $|\Prop|$ is finite there must be $j < k$ such that $q_j = q_k =: q$.
Then we have $q \in R_j$ and $q \in R_k$ and $S_j, S_k \in {[q]}_\Gamma$.
Now consider the sequence of substitutions $R$s:
\begin{center}
\begin{tabular}{rcl}
  $q$ occurs in $R_j$\\
  $q$ does not occur in $R_{j+1}$ & = & ${[k_{j+1} \col q \mapsto S_j]}^{\iota_j} R_j$ \\[1ex]
  \vdots & & \begin{minipage}{8cm}
    $q$ must be reintroduced by some step ${[k_m \col q_m \mapsto S_m]}^{\iota_m}$
    \end{minipage}\\
  $q$ occurs in $R_k$\\
  $q$ does not occur in $R_{k+1}$ & = & ${[k_{k+1} \col q \mapsto S_k]}^{\iota_k} R_k$ \\
\end{tabular}
\end{center}
In particular, because these substitutions happen along the same branch, $k_m \col q_m$ is an occurrence in $Q_j$.
Hence, reasoning inside $\Gamma$ we have
$q \equiv Q_j$
and
$q \equiv {[k_m \col q_m \mapsto Q_m]}^{\iota_m} Q_j$.
But because $q$ occurs in $Q_m$, the latter is a circular formula in $\Gamma$.
Contradiction!
\end{proof}

To illustrate our proof method, we give two examples.

\begin{example}
The chain of $Q_i$ might for example be all right parts of consequences in Example~\ref{ex:growNonCirc}.
That is, in the set ${[s]}_\Gamma$ we have
  $Q_0 = p$,
  $Q_1 = (q \land r)$,
  $Q_2 = ((p \land r) \land r)$,
  $Q_3 = (((q \land r) \land r) \land r)$,
and so on.
The sequence of occurrence substitutions is then
${[1 \col p \mapsto (q \land r)]}^0$,
${[1 \col q \mapsto (p \land r)]}^1$,
${[1 \col p \mapsto (q \land r)]}^2$,
and so on.
This is an easy case: each step replaces an occurrence within the previous substituens.
Hence the tree of substitutions shown in Figure~\ref{fig:subTreeLine} only has one branch.

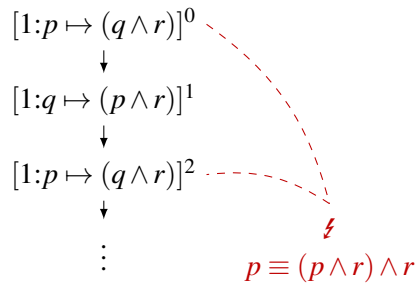
\begin{figure}[hb]
  \centering
  \begin{tikzpicture}[>=latex]
    \node (p1)               {$[1 \col p \mapsto (q \land r)]^0$};
    \node (q1) [below of=p1] {$[1 \col q \mapsto (p \land r)]^1$};
    \node (p2) [below of=q1] {$[1 \col p \mapsto (q \land r)]^2$};
    \node (dots) [below of=p2] {\vdots};
    \draw [->] (p1) -- (q1);
    \draw [->] (q1) -- (p2);
    \draw [->] (p2) -- (dots);
    \node (light) [myRed,right of=dots,node distance=3cm]
          {$\begin{array}{c}
              \text{\Lightning}\\
              p \equiv (p \land r) \land r
            \end{array}$};
    \draw [myRed,dashed] (p1.east) edge [bend left=20] (light.north);
    \draw [myRed,dashed] (light.north) edge [bend right=20] (p2.east);
  \end{tikzpicture}
  \caption{A simple substitution tree and the resulting circular formula.}\label{fig:subTreeLine}
\end{figure}
\end{example}

\begin{example}
An example which yields a proper tree is shown in Figure~\ref{fig:subTreeExample}.
On the left side we first show the $P_i$ chain.
Note that its formulas are strictly increasing in complexity, but for example $P_2$ and $P_3$ are not directly related via substitutions.
The second chain $Q_i$ is obtained using the $\mrg$ function from Definition~\ref{def:mrg} and restores this property.
For example, we can go from $Q_2$ to $Q_3$ via the substitution ${[2 \col r \mapsto \lnot \lnot s]}^2$.
All those substitutions are then arranged in the $(N,\rightsquigarrow)$ tree.
Finally, we show how two substitutions of the same proposition ($r$) in the same branch lead to a circular formula ($r \equiv \lnot \lnot (r \imp p)$).

\begin{figure*}
  \centering
  \begin{tikzpicture}[>=latex, node distance=0.8cm]
    \node (P0) [text width=4.5cm] {$P_0 = p$};
    \node (P1) [below of=P0, text width=4.5cm] {$P_1 = q \land r$};
    \node (P2) [below of=P1, text width=4.5cm] {$P_2 = (u \lor r) \land r$};
    \node (P3) [below of=P2, text width=4.5cm] {$P_3 = q \land \lnot \lnot s$};
    \node (P4) [below of=P3, text width=4.5cm] {$P_4 = (u \lor \lnot \lnot s) \land \lnot \lnot (r \imp p) $};
    \node (P5) [below of=P4, text width=4.5cm] {$P_5 = q \land \lnot \lnot (\lnot \lnot s \imp p)$};
    \node (PP) [below of=P5] {\vdots};
    \node (Q0) [right of=P0, text width=5cm, node distance=5cm] {$Q_0 = p$};
    \node (Q1) [below of=Q0, text width=5cm] {$Q_1 = q \land r$};
    \node (Q2) [below of=Q1, text width=5cm] {$Q_2 = (u \lor r) \land r$};
    \node (Q3) [below of=Q2, text width=5cm] {$Q_3 = (u \lor r) \land \lnot \lnot s$};
    \node (Q4) [below of=Q3, text width=5cm] {$Q_4 = (u \lor \lnot \lnot s) \land \lnot \lnot (r \imp p) $};
    \node (Q5) [below of=Q4, text width=5cm] {$Q_5 = (u \lor \lnot \lnot s) \land \lnot \lnot (\lnot \lnot s \imp p)$};
    \node (QQ) [below of=Q5] {\vdots};
    \node (1) [above of=Q1, node distance=0.5cm] {};
    \node (2) [above of=Q2, node distance=0.5cm] {};
    \node (3) [above of=Q3, node distance=0.5cm] {};
    \node (4) [above of=Q4, node distance=0.5cm] {};
    \node (5) [above of=Q5, node distance=0.5cm] {};
    \node (p)    [right of=1, node distance=5.5cm] {$[1 \col p \mapsto (q \land r)]^0$};
    \node (pq)   [right of=2, node distance=4cm] {$[1 \col q \mapsto (u \lor r)]^1$};
    \node (pr)   [right of=3, node distance=7cm] {$[2 \col r \mapsto (\lnot \lnot s)]^2$};
    \node (pqr)  [right of=4, node distance=4cm] {$[1 \col r \mapsto \lnot \lnot s]^3$};
    \node (prs)  [right of=4, node distance=7cm] {$[1 \col s \mapsto (r \imp p)]^3$};
    \node (prsr) [right of=5, node distance=7cm] {$[1 \col r \mapsto \lnot\lnot s]^4$};
    \node [below of=pqr] {\vdots};
    \node [below of=prsr] {\vdots};
    \draw [->] (p) -- (pq);
    \draw [->] (pq) -- (pqr);
    \draw [->] (p) -- (pr);
    \draw [->] (pr) -- (prs);
    \draw [->] (prs) -- (prsr);
    \node (light) [myRed,left of=prsr,below of=prsr, node distance=1.6cm]
          {$\begin{array}{c}
              \text{\Lightning}\\
              r \equiv \lnot\lnot (r \imp p)
            \end{array}$};
    \draw [myRed,dashed] (pr.west) edge [bend right=20] (light.north);
    \draw [myRed,dashed] (light.north) edge [bend left=20] (prsr.west);
  \end{tikzpicture}
  \caption{A definition chain and the resulting occurrence substitution tree. The right branch yields a circular formula.}\label{fig:subTreeExample}
\end{figure*}
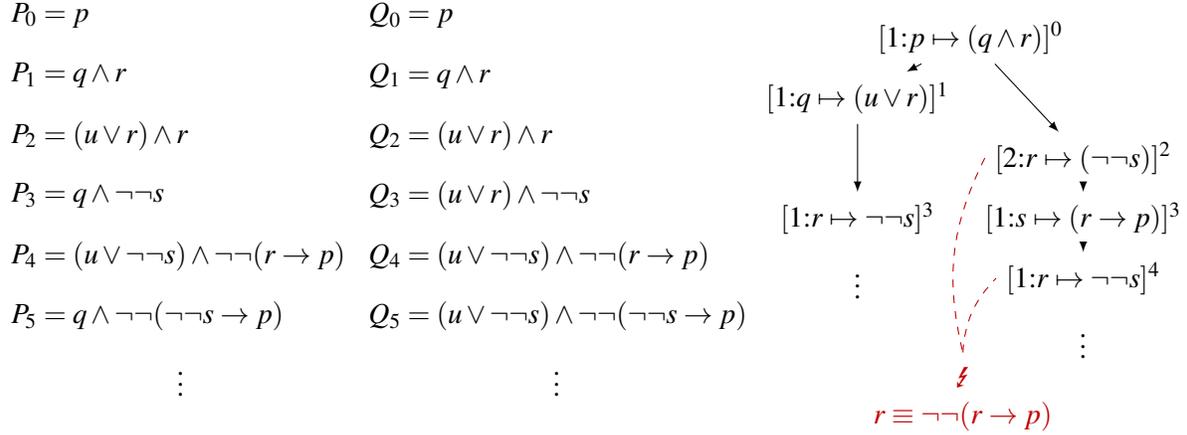
\end{example}

Given that all ${[p]}_\Gamma$ sets are finite, we can now choose a definition for each $p \in \Prop$,
namely the longest and among those the lexicographically least formula in ${[p]}_\Gamma$.
In fact, this is the same as applying all possible substitutions until reaching the leaves in the substitution trees from the proof of Lemma~\ref{lem:finiteDefSet}.

\begin{definition}
For any finite set of boolean formulas $X \subseteq \Lng_B$, we define $\pick(X)$ as the longest element of $X$ and among those the lexicographically first.

Formally, let $\pick(X) := \min_< \{ P \in X \mid \forall Q \in X : l(P) \geq l(Q) \}$
\end{definition}

\begin{definition}[Canonical Model]\label{def:canonicalm}
For a finite vocabulary $\Prop$ the corresponding canonical model $\M=(W,R,V,\DF)$ is defined as follows:
\begin{itemize}
\item $W := \{ \Gamma \subseteq \Lng(\Prop) \mid \Gamma \text{ is maximally consistent} \}$
\item $\Gamma_1 R_i \Gamma_2 :\iff  \{ \phi \mid \Box_i \phi \in \Gamma_1 \} \subseteq \Gamma_2$
\item $V(\Gamma) := \Gamma \cap \Prop$
\item For each $\Gamma$ and any $p \in \Prop$, let $\DF_\Gamma(p) := \pick({[p]}_\Gamma)$
\end{itemize}
\end{definition}

\begin{lemma}
The canonical model is indeed a model, and not just a premodel.
\end{lemma}
\begin{proof}
We need to show two properties.
\begin{itemize}
  \item For all $P,Q \in \Lng_B$:
    If $\df_w(P)=\df_w(Q)$, then we also have $V_w(P) = V_w(Q)$.
    This follows from the equivalence axiom $(P \equiv Q) \to (P \leftrightarrow Q)$ and the boolean part of the Truth Lemma~\ref{lem:truth} below which can be shown independently.
  \item For all $p,q \in \Prop$, we need to show that if $p$ occurs in $\DF_w(q)$, then $\DF_w(p) = p$.
    Now, take any $p$ in $\DF_w(q)$.
    Because $\DF_w(q)$ is among the longest formulas in ${[q]}_\Gamma$, there cannot be any $P \in {[p]}_\Gamma$ with $l(P) > 1$.
    Namely, ${[p]}_\Gamma$ only consists of propositional letters.

    Because $\DF_w(q)$ is the lexicographically first among the longest formulas in ${[q]}_\Gamma$, also $p$ must be lexicographically first in ${[p]}_\Gamma$, for otherwise we can replace it in $\DF_w(q)$ by another propositional letter which is lexicographically smaller than $p$ to obtain a lexicographically smaller formula in ${[q]}_\Gamma$.
    Together, we have $\DF_w(p) = p$.\qedhere
\end{itemize}
\end{proof}

\begin{lemma}\label{lem:pickExtended}
In the canonical model we have for all $P \in \Lng_B(\Prop)$ that $\df_\Gamma(P) = \pick({[P]}_\Gamma)$.
\end{lemma}
\begin{proof}
  By induction on the structure of $P$.
  The base case $P = p$ follows from Definition~\ref{def:canonicalm}.

  For the induction step, consider the case $P = \lnot Q$.
  Then we have the following chain of identities:
  \[ \df_\Gamma(\lnot Q)    \stackrel{\text{Def.~\ref{def:canonicalm}}}{=}
  \lnot \df_\Gamma(Q)       \stackrel{\text{IH}}{=}
  \lnot \pick({[Q]}_\Gamma) \stackrel{\ast}{=}
  \pick({[\lnot Q]}_\Gamma) \]
  where the step IH is by induction hypothesis and the step $\ast$ is shown as follows.
  By definition of $\pick$ we have that $\pick({[\lnot Q]}_\Gamma)$ is
\[ {\min}_< \{ P \in {[\lnot Q]}_\Gamma \mid \forall R \in {[\lnot Q]}_\Gamma : l(P) \geq l(R) \} \]
which by pattern matching is the same as
\[ {\min}_< \{ \lnot P \mid P \in {[Q]}_\Gamma \text{ and } \forall R \in {[Q]}_\Gamma : l(\lnot P) \geq l(\lnot R) \} . \]
Because $<$ and $\geq$ with respect to $l(\cdot)$ is preserved under negation, this is the same as
\[ \lnot {\min}_< \{ P \in {[Q]}_\Gamma \mid \forall R \in {[Q]}_\Gamma : l(P) \geq l(R) \} \]
which is $\lnot \pick({[Q]}_\Gamma)$.

A similar chain covers the case $P = (Q_1 \land Q_2)$.
\end{proof}

\begin{lemma}[Truth Lemma]\label{lem:truth}
Consider the canonical model $\M$.
For all worlds $\Gamma$ and all formulas $\phi$ without announcement operators we have
  $\M, \Gamma \vDash \phi$ iff $\phi \in \Gamma$.
\end{lemma}
\begin{proof}
By induction on the complexity of $\phi$.
The only non-standard case is the $\equiv$ operator.
We want to show
\[ P \equiv Q \ \in \Gamma \iff \M,\Gamma \vDash P \equiv Q \]
for which it suffices to show
\[ {[P]}_\Gamma = {[Q]}_\Gamma \iff \df_\Gamma(P) = \df_\Gamma(Q). \]

\emph{For left to right}, suppose ${[P]}_\Gamma = {[Q]}_\Gamma$.
This implies $\pick({[P]}_\Gamma) = \pick({[Q]}_\Gamma)$.
Hence we have $\df_\Gamma(P) = \df_\Gamma(Q)$ by Lemma~\ref{lem:pickExtended}.

\emph{For right to left}, suppose we have $\df_\Gamma(P) = \df_\Gamma(Q)$.
Then by Lemma~\ref{lem:pickExtended} we have a single formula $R := \pick({[P]}_\Gamma) = \pick({[Q]}_\Gamma)$.
In particular we have $R \in {[P]}_\Gamma$ and $R \in {[Q]}_\Gamma$.
Hence $P \equiv R \in \Gamma$ and $R \equiv Q \in \Gamma$.
Now by transitivity from Definition~\ref{def:axioms} we have $P \equiv Q \in \Gamma$.
Therefore ${[P]}_\Gamma = {[Q]}_\Gamma$.
\end{proof}

Finally, we can state and prove completeness.

\begin{theorem}[Completeness]
$\SMLD$ is weakly complete for announcement-free fragment of $\Lng$, and
$\SPALD$ is weakly complete for the full $\Lng$.
\end{theorem}
\begin{proof}
Suppose $\phi$ is not provable.
By the PAL reduction axioms there is a formula $\phi'$ without announcements such that $\vdash \phi \leftrightarrow \phi'$.
Hence $\phi'$ is not provable and therefore $\lnot \phi'$ is consistent.

Let $\Prop$ be the vocabulary of $\phi'$.
Let $\M$ be the canonical model for $\Prop$ per Definition~\ref{def:canonicalm}.
Let $\Gamma$ be any maximally consistent set containing $\lnot \phi'$.
Such a set always exists and can be defined using a standard Lindenbaum Lemma~\cite[p.~197]{BRV}.
In particular, $\Gamma$ is an element of $W$ in $\M$.
Then by the Truth Lemma we have $\M,\Gamma \vDash \lnot \phi'$.

The reduction axioms are also semantically valid, so we have $\M,\Gamma \vDash \lnot \phi$.
Hence $\phi$ is not semantically valid.
\end{proof}

\section{Knowing \emph{the} Definition}

Our language does not allow us to express that an agent knows \emph{the} definition of a proposition.
We can add this using an operator similar to $Kv$ from~\cite{WangFan2013KvPAL}.
Formally, let $\Kd_i P$ where $P \in \Lng_B$ have the following semantics:
\[ \M,w \vDash \Kd_i(P)  \ \iff \  \forall w' :
  w R_i w' \text{ implies } \df_w(P) = \df_{w'}(P) \]

As we have shown in Examples~\ref{ex:knowing-without-understanding} and~\ref{ex:understanding-without-knowing} knowing whether something is true and knowing its definition are contingent, neither implies the other.

We can then also define the notion of \emph{explicitly knowing},
which is the combination of knowing that and knowing the definition:
  \[ \Kx_i(P) :=  \Box_i P \land \Kd_i(P) \]
However, in our framework it is only possible to know propositional formulas explicitly.
For example, the formula $\Kx_i(\Box_j(p \land q))$ is not in the language.

Also adding \emph{the} definition itself to the language, with the following operator $:=$ could be helpful.
\[ \M,w \vDash p:=P  \ \iff \  \DF_w(p) = P \]
For example we then have the following validities:
\begin{itemize}
\item $p := P \imp p \equiv P$ \ (definition implies equivalence)
\item $p := P \imp \lnot (p := Q)$ for all $P \neq Q$ \ (uniqueness)
\item $(p := P \land \Kd_i(p)) \imp \Box_i (p := P)$
\item $\Kd_i P \land \Box_i(P \equiv Q) \to \Kd_i Q$
\end{itemize}
In fact, the $:=$ operator could also simplify our original completeness proof.
Choosing definitions for the canonical model is trivial for this extended language, because each maximally consistent set $\Gamma$ will contain exactly one formula of the form $p := P$ for each $p \in \Prop$.%
\footnote{We thank Alexandru Baltag for pointing out how $:=$ may simplify our logic.}

We leave it as future work to axiomatize the extension of our logic with $\Kd$.

\section{Conclusion and Future Work}

We presented an extension of Public Announcement Logic (PAL) to model the knowledge of meanings with boolean definitions.
In our logic agents can understand a proposition without knowing its truth value or the other way round.
Moreover, multiple agents can agree on something without agreeing on its meaning and vice versa.

We also presented a sound and complete axiomatization with intuitive axioms to characterize the equivalence operator $\equiv$.
The completeness proof is based on a standard canonical model construction, extended with two new ideas for boolean definitions.
We use $\mrg$ to combine different possible definitions and then use a tree of \emph{occurrence substitutions} to ensure that there are only finitely many definitions to choose from.

Our formal contributions are thus more about pattern matching and less the epistemic and dynamic operators.
Nevertheless, we think that our logic showcases an interesting interaction between boolean definitions and the standard operators $K$ and $[\phi]$.
On the other hand, as a subsystem of our proof system, the pattern matching logic regarding $\equiv$ seems interesting on its own.
It may be applied in computer science or combined with other philosophical logics.
In fact, our pattern matching and mismatching axioms are reminiscent of term rewriting rules.
We therefore conjecture that our central Lemma~\ref{lem:finiteDefSet} can also be shown using Kruskal's Theorem applied to term rewriting, as discussed in~\cite{Dershowitz1982:OrderingsRewriting}.%
\footnote{We thank one of the anonymous reviewers for suggesting this connection.}

Our work can be extended in several ways.

The framework is compatible with the range of multi-agent epistemic logics from $K_n$ to ${S5}_n$.
The usual axioms for frame properties can be added, for example one might add transitivity for positive introspection, or reflexivity for truthfulness of knowledge.

We already mentioned in the last section that our logic relates to the $\Kv$ operator from~\cite{WangFan2013KvPAL}.
We think that ``knowing value'' models could also be equipped with definitions for their general variables instead of propositions.
This can then be combined with ``public inspection'' from~\cite{vEGW2017:KvPIL} and the resulting framework might give a new perspective on knowledge of terms.

The distinction between value and meaning is best illustrated in the setting of cryptographic protocols, where you may know the value of the message but only part of the meaning.
To illustrate this, suppose you have your own private key $k$.
If you receive an encrypted message ${\{a, {\{b\}}_{k'}\}}_k$, then you can decode the outer level and learn the value of $a$.
But you cannot learn $b$ if it is encrypted by another key $k'$ that you do not have.
In fact, you might not even know that the message you received is of this form and only understand ${\{a, x\}}_k$.
Such phenomena also happen in everyday communication.
You may only get part of the meaning of a sentence, but by asking ``what do you mean by \dots?'' you can learn more.

After seeing the details of our framework one might wonder how it relates to other logics of ambiguity.
We only mention two related works and leave more detailed comparisons for the future.

A semantic approach is~\cite{HalpernKets14:LogicAmbig} where agents are given different valuation functions to encode their disagreement about (ambiguous) atomic propositions.
In our models with definitions there is no need for additional valuation functions and our approach is more syntactic.
In our logic, knowing the meaning of a proposition is not reducible to knowing the valuation, e.g.,  two tautologies can still be very different definitionally.
Moreover, we also handle uncertainty about the actual meaning of propositions and provide agents with a way to learn it by update and pattern matching.

Another related approach to model syntactic ambiguity is~\cite{Kuijer13:Sequent} where the meaning of connectives is not fixed.
For example, an agent might wonder whether $p \ast q$ means $p \lor q$ or $p \land q$.
A similar example could be modelled in our logic with $r \equiv (p \lor q)$ vs.\ $r \equiv (p \land q)$.

Finally, there are two obvious limitations of our logic.
First, all definitions are boolean but in principle also modal definitions like $p := \Box_j q$ are interesting.
Second, as mentioned above, logical equivalence and equivalence by definition are not connected.
For example, $\lnot (p \land q)$ and $\lnot p \lor \lnot q$ are logically equivalent, but due to the pattern mismatch axiom we can never have $\lnot (p \land q) \equiv (\lnot p \lor \lnot q)$.
Depending on the application, users of our logic might consider this a problem or a feature.
We think that our ideas can be extended in both directions, but leave this as future work.

\bibliographystyle{eptcs}
\bibliography{paldefrefs}

\end{document}